\documentclass[runningheads]{llncs}

\usepackage[T1]{fontenc}

\usepackage{graphicx}

\usepackage{makeidx}  
\usepackage{graphicx}
\usepackage{color}
\usepackage{lineno,hyperref}

\usepackage{amsfonts}
\usepackage{amsmath}
\usepackage{amsthm} 
\usepackage{enumerate}
\usepackage{makecell}
\usepackage{appendix}

\usepackage{mathrsfs}
\usepackage{cases}
\usepackage{xcolor}

\usepackage{tcolorbox}

\usepackage{fancyhdr}
\usepackage{ifthen}
\usepackage{todonotes}

\begin{document}

\title{Deterministic Single Exponential Time Algorithms for Co-Path Packing and Co-Path Set Parameterized by Treewidth}

\author{
    Yuxi Liu\orcidID{0009-0009-2171-9042}
    \and
    Kangyi Tian
    \and
    Mingyu Xiao\orcidID{0000-0002-1012-2373}
}

\authorrunning{Y. Liu et al.}

\institute{University of Electronic Science and Technology of China, Chengdu, China
\email{yuxiliu823@gmail.com, kangyitian947@gmail.com, myxiao@uestc.edu.cn}}

\maketitle              

\begin{abstract}
    The \textsc{Co-Path Packing} (resp., \textsc{Co-Path Set}) problem asks whether a given graph can be edited to a collection of induced paths by deleting at most $k$ vertices (resp., $k$ edges). Both are fundamental problems with significant applications in bioinformatics and have been extensively studied within the framework of exact and parameterized algorithms. Currently, the state-of-the-art approach utilizes the randomized ``Cut \& Count'' technique, which solves \textsc{Co-Path Set} in $O^*(4^{\mathbf{tw}})$ time and \textsc{Co-Path Packing} in $O^*(5^{\mathbf{pw}})$ time, where $\mathbf{tw}$ is treewidth and $\mathbf{pw}$ is pathwidth. However, as there is no known method to derandomize the ``Cut \& Count'' technique, the existence of deterministic single exponential time algorithms for these problems parameterized by treewidth has remained an open question. In this paper, we resolve this gap by providing deterministic single exponential time algorithms for both problems when parameterized by treewidth.

\end{abstract}

\section{Introduction}

In this paper, we investigate the parameterized versions of the \textsc{Co-Path Packing} and \textsc{Co-Path Set} problems.
Given an input graph $G$ and a parameter $k$, the \textsc{co-Path Packing} problem (resp. the \textsc{Co-Path set} problem) asks whether we can delete at most $k$ vertices (resp. edges) from the input graph such that the remaining graph is a collection of induced paths.
The formal definitions of these two problems are as follows.

\begin{tcolorbox}
    
\textsc{Co-Path Packing}\\
\textbf{Instance:} A graph $G=(V, E)$ and an integer $k$.\\
\textbf{Question:} Is there a vertex subset $S\subseteq V$ of size at most $k$ whose deletion makes the graph a collection of induced paths?\\
\vspace{-5mm}
\end{tcolorbox}

\begin{tcolorbox}
    
\textsc{Co-Path Set}\\
\textbf{Instance:} A graph $G=(V, E)$ and an integer $k$.\\
\textbf{Question:} Is there an edge subset $S\subseteq E$ of size at most $k$ whose deletion makes the graph a collection of induced paths?\\
\vspace{-5mm}
\end{tcolorbox}

The \textsc{Co-Path Packing} problem finds significant application in bioinformatics~\cite{DBLP:journals/ploscb/ChauveT08}.
Similarly, the \textsc{Co-Path Set} problem is naturally motivated by the challenge of ordering genetic markers in DNA. 
This task arises in radiation hybrid mapping, a technique that utilizes fragment data generated by breaking chromosomes with gamma radiation to determine the ordering~\cite{cox1990radiation,richard1991radiation,slonim1997building,sullivan2017fast}.

\vspace{2mm}
\noindent\textbf{Related work.}
In this paper, we primarily investigate parameterized algorithms.
When parameterized by the solution size $k$,  both \textsc{Co-Path Packing} and \textsc{Co-Path Set} have been studied extensively.
Regarding \textsc{Co-Path Packing},
Chen et al.~\cite{chen2010linear} initially established that this problem can be solved in $O^*(3.24^k)$ time.
Subsequently, Feng et al.~\cite{feng2015randomized} improved this bound using a randomized algorithm running in $O^*(3^k)$ time.
Tsur~\cite{tsur2022faster} provided $O^*(3^k)$-time algorithms solving this problem deterministically.
The current best randomized algorithm runs in $O^*(2.93^k)$ time, recently established by Liu and Xiao \cite{liu2026solving}.

Regarding \textsc{Co-Path Set},
Fernau~\cite{DBLP:journals/dam/Fernau08} initially established that this problem can be solved in $O^*(2.47^k)$ time.
This was subsequently improved by Zhang et al.~\cite{DBLP:journals/jco/ZhangJZ14}, who provided a deterministic $O^*(2.45^k)$-time algorithm.
Then, Feng et al.~\cite{DBLP:conf/faw/FengZL14,DBLP:journals/jco/FengZW16} proposed a randomized $O^*(2.17^k)$-time algorithm for \textsc{Co-Path Set}.
However, Sullivan and van der Poel~\cite{sullivan2017fast} identified a flaw in their analysis which invalidates their probability of a correct solution in the given time.
Also, Sullivan and van der Poel~\cite{sullivan2017fast} developed a randomized $O^*(1.59^k)$-time algorithm for \textsc{Co-Path Set}. 
Tsur~\cite{DBLP:journals/ipl/Tsur23a} provided $O^*(2^k)$-time algorithms solving this problem deterministically.

While deterministic single exponential algorithms (with respect to the solution size $k$) can be straightforwardly derived by branching and search method, obtaining such algorithms for the parameter treewidth ($\mathbf{tw}$) is significantly more challenging.
This difficulty arises from the global connectivity constraints for \textsc{Co-Path Packing} and \textsc{Co-Path Set}.
The problem with global connectivity constraint is also called \emph{connectivity-type problem}.
For connectivity-type problems, the typical dynamic programming approach must track all possible ways the solution can traverse the corresponding separator of the tree decomposition, leading to a state space $\Omega(\mathbf{tw}^\mathbf{tw})$~\cite{cygan2011solving}.
This issue was addressed in the landmark work of Cygan et al.~\cite{cygan2011solving}, who introduced the cut \& count technique to obtain randomized single exponential algorithms for various connectivity-type problems. 
Applying this technique, An $O^*(4^{\mathbf{tw}})$-time randomized algorithm was obtained for \textsc{Co-Path Set}~\cite{sullivan2017fast}, and an $O^*(5^{\mathbf{pw}})$-time randomized algorithm was obtained for \textsc{Co-Path Packing}~\cite{liu2026solving}.

However, the derandomization of cut \& count remains a major open problem.
To obtain deterministic single exponential time algorithms, 
alternative techniques such as matrix-based approaches~\cite{DBLP:journals/iandc/BodlaenderCKN15} and representative families~\cite{DBLP:journals/jacm/FominLPS16} have been proposed.
In this paper, we use representative families to obtain first deterministic single exponential time algorithms for both \textsc{Co-Path Packing} and \textsc{Co-Path Set}.

\vspace{2mm}
\noindent\textbf{Our contributions.}
The primary contribution of this paper is the deterministic single exponential algorithms for \textsc{Co-Path Packing} and \textsc{Co-Path Set} parameterized by treewidth.
Specifically, we provide:
\begin{enumerate}
    \item A deterministic algorithm running in $O((13 + 3\cdot 2^{\omega + 1})^{\mathbf{tw}} \mathbf{tw}^{O(1)}n)$ for \textsc{Co-Path Packing}.
    \item A deterministic algorithm running in $O((12 + 3 \cdot 2^{\omega + 1})^{\mathbf{tw}} \mathbf{tw}^{O(1)}n)$ time for \textsc{Co-Path Set}.
\end{enumerate}
Furthermore, when parameterized by pathwidth, we provide:
\begin{enumerate}
    \item A deterministic algorithm running in $O((2 + 2^{\omega + 1})^{\mathbf{pw}}\mathbf{pw}^{O(1)}n)$ time for \textsc{Co-Path Packing}.
    \item A deterministic algorithm running in $O((1 + 2^{\omega + 1})^{\mathbf{pw}}\mathbf{pw}^{O(1)}n)$ time for \textsc{Co-Path Set}.
\end{enumerate}

As the algorithms for \textsc{Co-Path Set} are slightly simpler than those for \textsc{Co-Path Packing}, we first present the algorithms for \textsc{Co-Path Set} in Section~3, and subsequently present the algorithms for \textsc{Co-Path Packing} in Section~4.

\section{Preliminaries}

\vspace{2mm}
\noindent\textbf{Graphs.}
In this paper, we only consider simple and undirected graphs.
Let $G=(V, E)$ be a graph with $n=|V|$ vertices and $m=|E|$ edges.
For a graph $G'$, we use $V(G')$ and $E(G')$ to denote its vertex and edge sets, respectively.

A vertex $v$ is called a \emph{neighbor} of a vertex $u$ if there is an edge $uv \in E$.
Let $N(v)$ denote the set of neighbors of $v$. For a vertex subset $X\subseteq V$, let $N(X)=\bigcup_{v\in X}N(v)\setminus X$ and $N[X]=N(X)\cup X$.
We use $d(v)=|N(v)|$ to denote the \emph{degree} of a vertex $v$ in $G$.
A vertex of degree $d$ is a \emph{degree-$d$ vertex}.
For a vertex subset $X\subseteq V$, the subgraph induced by $X$ is denoted by $G[X]$. 
The induced subgraph $G[V\setminus X]$ is also written as $G\setminus X$.
For an edge subset $D \subseteq E$, let $V(D)$ denote the set of vertices incident to at least one edge in $D$. 
The subgraph $(V(D), D)$ is denoted by $G[D]$ and $G - D = (V, E\setminus D)$ denotes the graph obtained by removing the edges in $D$ from $G$.

A \emph{path} $P$ in $G$ is a sequence of vertices $v_1, v_2,\cdots, v_t$ such that $v_{i}v_{i+1} \in E$ for $1\leq i < t$.
A \emph{cycle} $C$ in $G$ is a sequence of vertices $v_1, v_2,\cdots, v_t$ such that $v_{i}v_{i+1} \in E$ for any $1\leq i < t$ and $v_1v_t\in E$.
Two vertices $u$ and $v$ are \emph{reachable} to each other if there is a path $v_1, v_2,\cdots, v_t$ such that $v_1 = u$ and $v_t = v$.
A graph is called \emph{connected}, if any two vertices in it are reachable to each other.
A \emph{connected component} of a graph is a maximal connected subgraph.
A connected graph is called a \emph{tree} if there is no cycle in it.
A graph is called a \emph{forest} if each connected component of this graph is a tree.
An \emph{induced path} in $G$ is a path $P$ such that $G[V(P)]$ has no cycles.

\vspace{2mm}
\noindent\textbf{Tree decomposition.} 
We employ the standard definitions of tree decompositions and nice tree decompositions.

\begin{definition}[\cite{cygan2015parameterized}]
    A \textit{tree decomposition} of a graph $G$ is a pair ${\cal T} = (T, \{X_t\}_{t\in V(T)})$, where $T$ is a tree whose every node $t$ is assigned a vertex subset $X_t\subseteq V(G)$ (called a \emph{bag}) such that:\\
    \emph{(T1)} $\bigcup_{t\in V(T)} X_t = V(G)$.\\
    \emph{(T2)} For every $uv \in E(G)$, there exists a node $t$ of $T$ such that $X_t$ contains both $u$ and $v$.\\
    \emph{(T3)} For every $u \in V(G)$, the set $T_u=\{t\in V(T) : u\in X_t\}$ induces a connected subtree of $T$.
\end{definition}

The \emph{width} of a tree decomposition $(T, \{X_t\}_{t\in V(T)})$ is $\max_{t\in V(T)} \{|X_t|\} - 1$.
The \emph{treewidth} of a graph $G$, is the minimum width of all possible tree decompositions of $G$.
To avoid ambiguity, we refer to the elements of $V(T)$ as \emph{nodes}.
For convenience, we consider $T$ to be a rooted tree with a designated root node $r$.

A tree decomposition ${\cal T} = (T, \{X_t\}_{t\in V(T)})$ is \textit{nice} if the following conditions are satisfied:
\begin{enumerate}
    \item For any leaf $l$ of $T$, $X_l=\emptyset$. We call it \emph{leaf node}.
    \item Every non-leaf node $T$ is one of the following three types:
        \begin{enumerate}
            \item \emph{Introduce node}: A node $t$ with exactly one child $t'$ such that $X_t = X_{t'} \cup \{v\}$ for some vertex $v\notin X_{t'}$;
            \item \emph{Forget node}: A node $t$ with exactly one child $t'$ such that $X_t = X_{t'} \setminus \{v\}$ for some vertex $v\in X_{t'}$;
            \item \emph{Join node}: A node $t$ with two children $t_1$ and $t_2$ such that $X_t=X_{t_1}=X_{t_2}$.
            
        \end{enumerate} 
\end{enumerate}

The following lemma establishes that any tree decomposition can be turned into a nice tree decomposition in linear time without increasing the width.

\begin{lemma}[\cite{cygan2015parameterized}]\label{Lemma:nice}
    If a graph $G$ admits a tree decomposition of width at most $k$, then it also admits a nice tree decomposition of width at most $k$.
    Moreover, given a tree decomposition ${\cal T} = (T, \{X_t\}_{t\in V(T)})$ of $G$ of width at most $k$, one can in time $O(k^2\cdot \max(|V(T)|, |V(G)|))$ compute a nice tree decomposition of $G$ of width at most $k$ that has at most $O(k|V(G)|)$ nodes.
\end{lemma}

A \emph{path decomposition} is a specialized tree decomposition where the tree $T$ is restricted to be a path (i.e., there is no join node.). 
The pathwidth of a graph is the minimum width of all possible path decompositions. 
Any path decomposition can be transformed into a \emph{nice path decomposition} in linear time.

\vspace{2mm}
\noindent\textbf{Representative Family.} Before we introduce the concepts of representative family, we introduce the concepts of matroids.

\begin{definition}[Matroid]
    A pair $M = (E, {\cal I})$, where $E$ is a ground set and $\cal I$ is a family of subsets (called independent sets) of $E$, is a \emph{matroid} if it satisfies the following conditions:\\
    \emph{(I1)} $\emptyset \in \mathcal{I}$.\\
    \emph{(I2)} If $A'\subseteq A$ and $A\in {\cal I}$, then $A'\in {\cal I}$\\
    \emph{(I3)} If $A, B\in {\cal I}$ and $|A| < |B|$, then there is $e\in (B\setminus A)$ such that $A\cup \{e\} \in {\cal I}$.
\end{definition}

Axiom (I3) implies that all the inclusion-wise maximal independent sets have the same size, which is called the \emph{rank} of the matroid $M$.
For a graph $G$, we further define a \emph{graphic matroid} $M = (E, {\cal I})$ by taking elements as edges of $G$ and $F\subseteq E(G)$ is in $\cal I$ if $G[F]$ is a forest.
For more details about matroids, please refer to~\cite{oxley2006matroid}.

\begin{definition}[Max $q$-Representative Family, \cite{DBLP:journals/jacm/FominLPS16}]
    Given a matroid $M = (E, {\cal I})$, a family ${\cal S}$ of subsets of $E$, and a nonnegative weight function $w : {\cal S} \to \mathbb{N}$, 
    we say that a subfamily $\hat {\cal S} \subseteq {\cal S}$ is max $q$-representative for $\cal S$ if the following holds: 
    for every set $Y \subseteq E$ of size at most $q$, if there is a set $X\in {\cal S}$ disjoint from $Y$ with $X\cup Y \in {\cal I}$, then there is a set $\hat X\in \hat {\cal S}$ disjoint from $Y$ with\\
    \emph{(1)} $\hat X\cup Y \in \mathcal{I}$; and\\
    \emph{(2)} $w(\hat X)\geq w(X)$.
\end{definition}
We use $\hat {\cal S}\subseteq^q_{maxrep} {\cal S}$ to denote a max $q$-representative family for $\cal S$.
The following lemma presents an algorithm to find such a max $q$-representative family. 
For convenience, we only consider the case that the matroid is a graphic matroid.

\begin{lemma}[\cite{DBLP:journals/jacm/FominLPS16}]\label{Lemma:RepComputation}
     Let $M = (E, {\cal I})$ be a graphic matroid of rank $p + q = k$, ${\cal S} = \{S_1,\dots,S_t\}$ be a family of independent sets that for each $S_i\in {\cal S}, |S_i| = p$, and $w : {\cal S} \to \mathbb{N}$ be a nonnegative weight function.
     Then there exists $\hat {\cal S}\subseteq^q_{maxrep} {\cal S}$ of size ${p + q\choose p}$. 
     Moreover, we can find $\hat {\cal S}\subseteq^q_{maxrep} {\cal S}$ of size at most ${p + q\choose p}$ in $O({p+q\choose q}tp^\omega + t{p+q\choose q}^{\omega - 1})$ time.
     Here, $\omega < 2.372$ is the matrix multiplication exponent.
     \cite{alman2025more}
\end{lemma}

\section{The algorithm for \textsc{Co-Path Set}}
In this section, we consider an optimal and weighted variant of \textsc{Co-Path Set}, called \textsc{Weighted Co-Path Set}. 
In this problem, we seek a maximum weight subgraph that is a collection of induced paths.
The formal definition is as follows.

\begin{tcolorbox}
\textsc{Weighted Co-Path Set}\\
\textbf{Instance:} A graph $G=(V, E)$ and a weight function $w : E(G) \to \mathbb{N}$.\\
\textbf{Task:} Find a subset of $E(G)$ of maximum weight whose induced graph is a collection of induced paths.\\
\vspace{-5mm}
\end{tcolorbox}

In this section, 
an edge subset $E'\subseteq E(G)$ is called a \emph{solution} if $G[E']$ is a collection of induced paths.
An edge subset $E'\subseteq E(G)$ is called an \emph{optimal solution} if $E'$ is a solution of the maximum weight.
Let $\mathscr{L}$ be a family of edge subsets that contains all optimal solutions.

Assuming there is a tree decomposition ${\cal T} = (T, \{X_t\}_{t\in V(T)})$ of $G$.
For any node $t$, we use $V_t$ to denote the vertex set $\bigcup_{t'} X_{t'}$, where $t'$ ranges over $t$ and all descendants of $t$. Let $G_t = (V_t, E(G[V_t]) \setminus E(X_t))$. 
We further use ${\cal S}_t$ to denote the family containing all edge subsets of $E(G_t)$ that induce a collection of induced paths.
We call ${\cal S}_t$ a \emph{family of partial solutions} of $t$. 
Then, for each partition $(R_0, R_1, R_2)$ of $X_t$, we define a family ${\cal S}_t[R_0, R_1, R_2]$ that contains partial solutions of $t$ as follows.

\begin{definition}
    Consider a node $t \in T$ and a partition $(R_0, R_1, R_2)$ of $X_t$. 
    We define $\mathcal{S}_t[R_0, R_1, R_2]$ as a family of all partial solutions $L_t \subseteq E(G_t)$ such that

        $R_i$ ($i = 0, 1, 2$) is a subset of $X_t$ containing degree-$i$ vertices in $G_t[L_t]$, respectively.
    
    \label{def-s_t-cPS}
\end{definition}

Recall that ${\cal S}_t$ is the family containing all edge subsets of $E(G_t)$ that induce a collection of induced paths.
Clearly, $\mathcal{S}_t = \bigcup_{(R_0, R_1, R_2) \text{ is a partition of } X_{t}}{\cal S}_{t}[R_0, R_1, R_2]$.
In our approach, 
rather than maintaining a standard DP table for each node $t$ of the tree decomposition,
We keep a family of partial solutions for the graph $G_t$. 
This family of partial solutions is constructed such that for every optimal solution $L\in \mathscr{L}$ and its corresponding intersection $L_t = E(G_t) \cap L$ with the graph $G_t$, there exists a partial solution $\hat L_t$ in the family such that $\hat L_t \cup (L\setminus L_t)$ is also an optimal solution. 
Specifically, for each node $t$ and each partition $(R_0, R_1, R_2)$ of $X_t$, we keep a family of edges $\hat {\cal S}_t[R_0, R_1, R_2]$ of $G_t$ satisfying the following correctness invariant.

\vspace{2mm}
\noindent
\textbf{Correctness Invariant.} For each $L\in \mathscr{L}$, let $L_t = L \cap E(G_t)$ and $L_R = L\setminus L_t$.
Let $(R_0, R_1, R_2)$ be the partition of $X_t$ such that $L_t\in \mathcal{S}_t[R_0, R_1, R_2]$.
Then, there exists $\hat L_t\in \hat {\cal S}_t[R_0, R_1, R_2]$ such that (i) $\hat L = \hat L_t \cup L_R \in \mathscr{L}$ and (ii) $\hat L_t\in {\cal S}_t[R_0, R_1, R_2]$.

\vspace{2mm}

Our dynamic programming algorithm traverses the nodes of the tree $T$ from the leaf nodes to the root node.
To show the correctness of the algorithm, we maintain this correctness invariant throughout this dynamic programming.
Then, we show that we can use representative family to obtain $\hat {\cal S}_t[R_0, R_1, R_2]$ of small size.
Specifically, we maintain the following size invariant:

\vspace{2mm}
\noindent
\textbf{Size Invariant.} For each node $t$ and each partition $(R_0, R_1, R_2)$, after the process of the algorithm, we have that $|\hat {\cal S}_t[R_0, R_1, R_2]|\leq 2^{|R_0| + |R_1|}$.
\vspace{2mm}

The difficulty in reducing ${\cal S}_t[R_0, R_1, R_2]$ to a single exponential size is that we need to consider all possible connectivity situations for the vertices in $R_0 \cup R_1$, of which there are $O((|R_0| + |R_1|)^{|R_0| + |R_1|})$.
However, by using the representative family technique,
we can only consider a single exponential size of connectivity situations for $R_0 \cup R_1$.
Firstly, we define the concept of \emph{connectivity graphs} to denote the connectivity situations of $R_0\cup R_1$.
Though vertices in $R_0$ are isolated vertices in the partial solution, they might be adjacent to some vertices not in $V_t$ in the optimal solution.
Thus, we should also consider $R_0$ in the connectivity graph.
For some edge set $E'\in E(G)$, let $w(E') = \sum_{e\in E} w(e)$.

\begin{definition}[Connectivity Graph]\label{def-con-graph-cPST}    
    Let $t$ be a node of $T$ and $E'\in E(G)$ be an edge set such that $G[E']$ is a collection of induced paths. 
    The \emph{connectivity graph} $F_t(E') = (V_F, E_F)$ of $t$ and $E'$ is obtained as follows:
    Let $R_1'$ be the set of vertices of degree 1 in $G[E']$ and let $R_1 = X_t \cap R_1'$. 
    
    Let $R_0 = X_t \setminus V(G[E'])$.
    Each vertex in $V_F$ corresponds to a vertex in the vertex set $R_0\cup R_1$. 
    Two vertices $u, v\in V_F$ are adjacent if and only if the vertices in $R_0\cup R_1$ corresponding to $u$ and $v$ belong to the same induced path in $G[E']$.
    We also define $w(F_t(E')) = w(E')$.
    
\end{definition}

\begin{figure}[!t]
    \centering
    \includegraphics[height = 5cm]{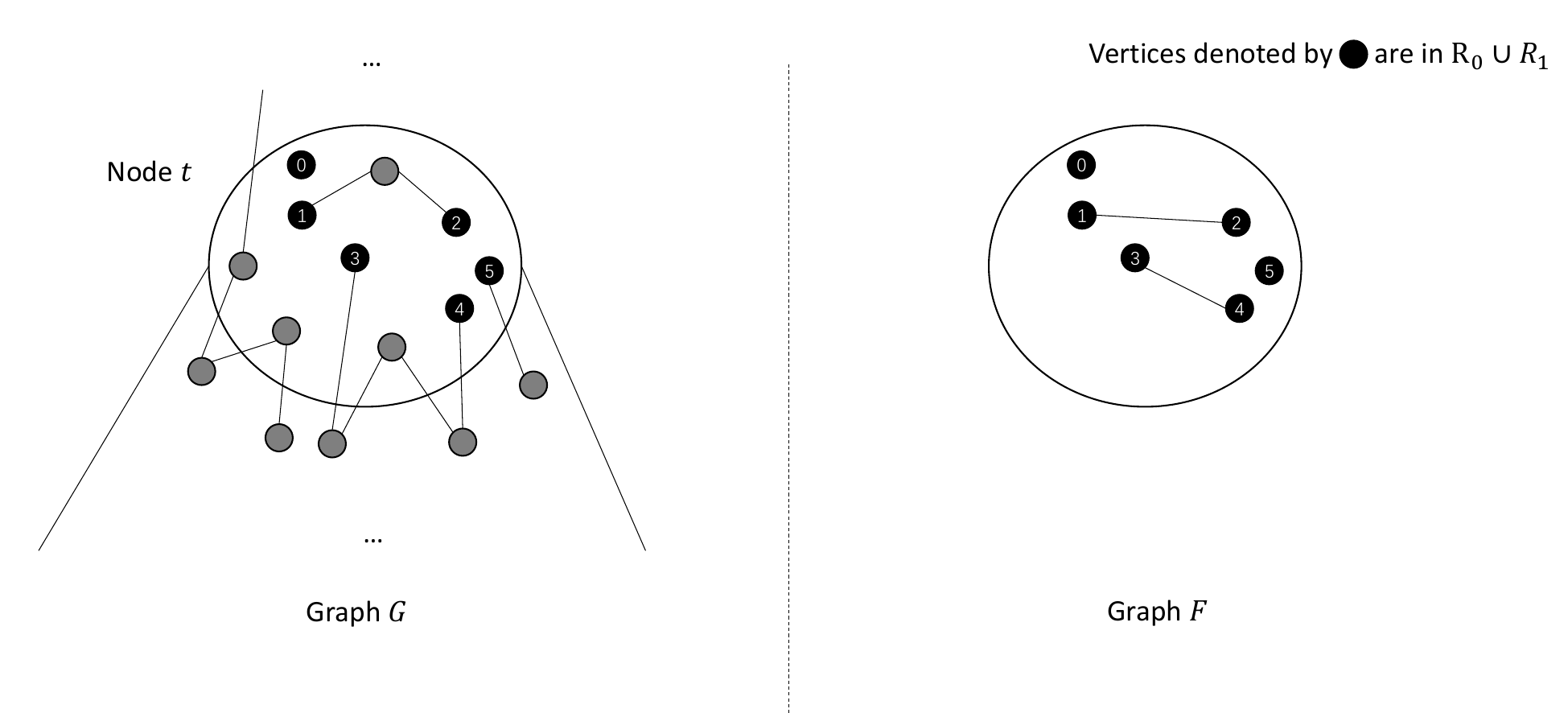}
    \caption{An example for a connectivity graph of $t$ and $E'$. 
    All edges marked with black lines in $G$ (the left graph) are in $E'$.
    The vertices marked black are in $V(F)$. The vertex marked 0 is in $R_0$ and the vertices marked 1-5 are in $R_1$. The vertices marked 1 and 2 (3 and 4) are adjacent since they are in the same induced path in $G[E']$.}
    \label{fig:ConnectivityGraph-cPST}
\end{figure}

Connectivity graphs are used to track which pairs of vertices in $R_1$ belong to the same induced path of $G[E']$.
However, it is not necessary that $E'$ is a partial solution.
See Fig. \ref{fig:ConnectivityGraph-cPST} for an illustration.
We use $F(E')$ to denote $F_t(E')$ when node $t$ is clear in context.
To show how to maintain the size invariant, we have the following lemma.

\begin{lemma}\label{Lemma:shrinking-cPST}
    Let $t$ be a node of $T$ and $(R_0, R_1, R_2)$ be a partition of $X_t$. 
    Let $\hat {\cal S}_t[R_0, R_1, R_2]$ be a family of edge subsets of $G_t$ satisfying the correctness invariant with $|\hat {\cal S}_t[R_0, R_1, R_2]| = l$ and $|R_0| + |R_1| = k$.
    Then, we can compute $\hat {\cal S}'_t[R_0, R_1, R_2]\subseteq \hat {\cal S}_t[R_0, R_1, R_2]$ satisfying the correctness invariant and size invariant in $O(2^{k(\omega - 1)}k^{O(1)}nl)$ time.
\end{lemma}

\begin{proof}
    Let $K[R_0\cup R_1]$ be a complete graph on the vertex set $R_0\cup R_1$.
    We consider a graphic matroid $M$ on $K[R_0\cup R_1]$.

    Let $\hat {\cal S}_t[R_0, R_1, R_2] = \{E_1^t, \dots, E_l^t\}$ and let ${\cal N} = \{F(E_1^t), \dots, F(E_l^t)\}$.
    For each $i\in \{1,\dots, k - 1\}$, let ${\cal N}_i$ be the family of a collection of paths of ${\cal N}$ with $i$ edges.
    For each ${\cal N}_i$ we apply Lemma \ref{Lemma:RepComputation} to compute its max $(k - 1 - i)$-representative family
    \[
        \hat {\cal N}_i\subseteq_{maxrep}^{k - 1 - i} {\cal N}_i.
    \]
    We now construct $\hat {\cal S}'_t[R_0, R_1, R_2]\subseteq \hat {\cal S}_t[R_0, R_1, R_2]$ as follows. We put every $E_j^t\in \hat {\cal S}_t[R_0, R_1, R_2]$ into $\hat {\cal S}'_t[R_0, R_1, R_2]$ if $F(E_j^t) \in \bigcup_{i = 1}^{k - 1}\hat{\cal N}_i$.
    By Lemma \ref{Lemma:RepComputation}, we have that $|\hat {\cal S}'_t[R_0, R_1, R_2]|\leq \sum_{i = 1}^{k - 1}{k - 1\choose i}\leq 2^{k} = 2^{|R_0|+|R_1|}$. Thus, the size invariant holds $\hat {\cal S}'_t[R_0, R_1, R_2]$.
    
    Next, we show that the correctness invariant holds.
    Consider any optimal solution $L\in \mathscr{L}$ and let $L_t = L \cap E_t$, $L_R = L\setminus L_t$.
    Let $R_i\subseteq X_t$ ($i = 0, 1, 2$) be the vertex set containing degree-$i$ vertices in $G_t[L]$. 
    Since $\hat {\cal S}_t[R_0, R_1, R_2]$ satisfies the correctness invariant,
    we have that there exists a partial solution $E_j^t\in \hat {\cal S}_t[R_0, R_1, R_2]$ such that $\hat L = E_j^t \cup L_R$ is an optimal solution and $R_i\subseteq X_t$ ($i = 0, 1, 2$) exactly contains degree-$i$ vertices in $G_t[\hat L]$.
    Since $\hat L$ is a feasible solution, we have that $F_t(E_j^t\cup L_R)$ is a collection of induced paths.
    
    Suppose that $F_t(E_j^t) \in \mathcal{N}_i$ for some $i\geq 0$.
    Since $\hat {\cal N}_i\subseteq_{maxrep}^{k - 1 - i} {\cal N}_i$, we have that there exists $F(E_h^t)\in \hat {\cal N}_i$ such that $w(F(E_h^t))\geq w(F(E_j^t))$ and there is no cycle in the graph $F_t(E_h^t\cup L_R)$. 
    Let $\hat L' = E_h^t\cup L_R$. 

    We then prove that $\hat L'$ is an optimal solution.
    First, we have that $w(\hat L') = w(L_R) + w(F(E_h^t)) \geq w(L_R) + w(F(E_j^t)) = w(\hat L)$.
    Since each degree of vertices in $V(G[E_j^t])\cap X_t$ is the same as the corresponding vertices in $V(G[E_h^t])\cap X_t$, there is no vertex of degree at least 3 in the graph $G[\hat L']$.
    Recall that $G[E_h^t]$ and $G[L_R]$ are both collections of induced paths.    
    Suppose that there exists a cycle in the graph $G[\hat L']$.
    This cycle must contain at least one induced path $P_1$ in $G[E_R]$ and at least one induced path $P_2$ in $G[E_h^t]$.
    Let $v_1$ and $v_2$ the two endpoints of $P_2$.
    Then $v_1$ and $v_2$ are both in $R_1$ and thus in $F_t(E_h^t\cup E_R)$.
    
    However, in this case,  $v_1$ and $v_2$ are contained in a cycle in the graph $F_t(E_h^t\cup L_R)$, which is a contradiction.
    Thus the graph $G[\hat L']$ does not contain any cycles.
    Therefore, $\hat L'$ is an optimal solution.

    By Lemma~\ref{Lemma:RepComputation}, the running time to compute $\hat {\cal S}'_t[R_0, R_1, R_2]$ is bounded by
    \[
        O\left(l\left(\sum_{i = 1}^{k - 1}{k - 1\choose i}k^\omega + \sum_{i = 1}^{k - 1}{k - 1\choose i}^{\omega - 1}\right)\right),
    \]
    which is bounded by
    \[
        O\left(l\left(\sum_{i = 1}^{k - 1}{k - 1\choose i}^{\omega - 1}k^\omega\right)\right)=O(2^{k(\omega - 1)}k^{O(1)}l).
    \]
    For a given edge set, we need to compute the corresponding connectivity graph and that takes $O(n)$ time.
    Thus, this lemma holds.
\end{proof}
 
In the dynamic programming process, the size of $\hat {\cal S}_t[R_0, R_1, R_2]$ can be larger than $2^{|R_0| + |R_1|}$.
Then we can use Lemma \ref{Lemma:shrinking-cPST} to reduce its size to at most $2^{|R_0| + |R_1|}$.
Now, we are ready to introduce our main algorithm in this section.

\begin{theorem}
    Given a tree decomposition of $G$ with width $\mathbf{tw}$. \textsc{Co-Path Set} can be solved in $O((12 + 3 \cdot 2^{\omega + 1})^{\mathbf{tw}} \mathbf{tw}^{O(1)}n)$ time.
    Given a path decomposition of $G$ with width $\mathbf{pw}$. \textsc{Co-Path Set} can be solved in $O((1 + 2^{\omega + 1})^{\mathbf{pw}}\mathbf{pw}^{O(1)}n)$ time.
\end{theorem}

\begin{proof}
    We can simply assume that the tree decomposition ${\cal T} = (T, \{X_t\}_{t\in V(T)})$ is a nice tree decomposition by Lemma~\ref{Lemma:nice}.
    We show how $\hat {\cal S}_t[R_0, R_1, R_2]$ is obtained for each node $t$ and each partition $(R_0, R_1, R_2)$ of $X_t$ from leaf nodes to the root node.
    For each node type, we first show the process of our algorithm and then show the maintenance of two invariants and time complexity.
    
\vspace{2mm}
\noindent
    \textbf{Leaf Node $t$.} It trivially holds that $\hat {\cal S}_t[R_0, R_1, R_2] = \emptyset$ for each node $t$ and each partition $(R_0, R_1, R_2)$ of $X_t$.
    Clearly, the correctness and size invariants are maintained.

\vspace{2mm}
\noindent
    \textbf{Introduce Node $t$ with Child $t'$.} Let $X_t = X_{t'}\cup \{v\}$ for some $v\notin X_{t'}$.
    By the definition of tree decomposition, we know that all neighbors of $v$ in $V_t$ must be in $X_{t'}$.
    Thus, by the definition of $G_t$, we know that $v$ is a degree-0 vertex in $G_t$.
    For every partition $(R_0, R_1, R_2)$ of $X_t$, one of the following two cases will happen.
    \begin{enumerate}
        \item $v\in R_1\cup R_2$. It holds that $\hat {\cal S}_t[R_0, R_1, R_2] =\emptyset$ since $v$ is a degree-0 vertex in $G_t$.
        Clearly, the correctness and size invariants are maintained.
        \item $v\in R_0$. It holds that $\hat {\cal S}_t[R_0, R_1, R_2] =\hat {\cal S}_{t'}[R_0\setminus \{v\}, R_1, R_2]$ since we have not added any edge to the newgraph.
        Since the family of solutions does not change, the correctness and size invariants are trivially maintained.
    \end{enumerate}

\vspace{2mm}
\noindent
    \textbf{Forget Node $t$ with Child $t'$.} Let $X_t = X_{t'}\setminus \{v\}$ for some $v\notin X_{t}$.
    Let $E_v[X_t]$ denote the set of edges between $v$ and the vertices in $X_t$. We have that $|E_v[X_t]| \leq \mathbf{tw}$.
    It is easy to see that $E(G_t) = E(G_{t'})\cup E_v[X_t]$.
    Thus, we need to enumerate every possible way an optimal solution can intersect with these newly added edges.
    Note that we can add at most two edges in $E_v[X_t]$ to partial solutions.
    Let $D\subseteq E_v[X_t]$ be the edges we add to partial solutions (it is possible that $D = \emptyset$).
    Consider each partition $(R_0, R_1, R_2)$ of $X_t$, 
    we further let ${\cal T}_{D}[R_0, R_1, R_2]$ be the family of partial solutions obtained by adding $D$ to partial solutions in ${\cal S}_t'$ such that $R_i$ ($i = 0, 1, 2$) exactly contains degree-$i$ vertices in the induced subgraph of each partial solutions in ${\cal T}_{D}[R_0, R_1, R_2]$.
    For $D = \emptyset$, it holds that
    \[
        \begin{split}
            {\cal T}_{\emptyset}[R_0, R_1, R_2] =  \{E': E'\in (&\hat {\cal S}_{t'}[R_0\cup \{v\}, R_1, R_2]\cup\\
                                                 &\hat {\cal S}_{t'}[R_0, R_1\cup \{v\}, R_2]\cup\\
                                                 &\hat {\cal S}_{t'}[R_0, R_1, R_2\cup \{v\}])\}.
        \end{split}
    \]

    For $D = \{vw\}$, we consider the following cases.
    \begin{enumerate}
        \item $w\in R_0$. This case is invaild since edge $vw$ is in the partial solution and $w$ should be in $R_1$ or $R_2$.
        Thus, it holds that ${\cal T}_{\{vw\}}[R_0, R_1, R_2] =\emptyset$.
        \item $w\in R_1$. It holds that 
        \[
            \begin{split}
                {\cal T}_{\{vw\}}[R_0, R_1, R_2] =  \{E'\cup \{vw\}: E'\in (
                    &\hat {\cal S}_{t'}[R_0\cup \{v,w\}, R_1\setminus \{w\}, R_2]\cup\\
                    &\hat {\cal S}_{t'}[R_0\cup \{w\}, (R_1\setminus \{w\}) \cup \{v\}, R_2])\}. 
            \end{split}
        \]
        \item $w\in R_2$. It holds that 
        \[
            \begin{split}
                    {\cal T}_{\{vw\}}[R_0, R_1, R_2] =  \{E'\cup \{vw\}: E'\in (
                        &\hat {\cal S}_{t'}[R_0\cup \{v\}, R_1\cup \{w\}, R_2\setminus \{w\}]\cup\\
                        &\hat {\cal S}_{t'}[R_0, R_1 \cup \{v, w\}, R_2\setminus \{w\}])\}. 
            \end{split}
        \]
    \end{enumerate}

    Similar, for $D = \{vw_1, vw_2\}$, we consider the following cases.
    \begin{enumerate}
        \item $w_1\in R_0$ or $w_2 \in R_0$. This case is invaild since edge $vw_1$ (resp. $vw_2$) is in the partial solution and $w_1$ (resp. $w_2$) should be in $R_1$ or $R_2$.
        Thus, it holds that ${\cal T}_{\{vw_1, vw_2\}}[R_0, R_1, R_2] =\emptyset$.
        \item $w_1\in R_1\cup R_2$ and $w_2\in R_1\cup R_2$. 
        Let $W = \{w_1, w_2\}$.
        Let $W_l = W\cap R_l$ for $l = 1, 2$. Let $W_0 = W_3 = \emptyset$.

        It holds that 
        \[
                {\cal T}_{\{vw_1, vw_2\}}[R_0, R_1, R_2] =  \{E'\cup \{vw_1, vw_2\}: E'\in (
                    \hat {\cal S}_{t'}[R_0'\cup \{v\}, R_1', R_2'])\}.
        \]
        where $R_l' = (R_l\setminus W_l) \cup W_{l + 1}$ for $l = 0, 1, 2$.
    \end{enumerate}

    Thus, we have that
    \[
        \begin{split}
            \hat {\cal S}_t[R_0, R_1, R_2] = 
                &{\cal T}_{\emptyset}[R_0, R_1, R_2]\\
                \cup & \bigl(\bigcup_{vw\in E_v[X_t]}{\cal T}_{\{vw\}}[R_0, R_1, R_2]\bigr)\\
                \cup &\bigl(\bigcup_{vw_1, vw_2\in E_v[X_t]}{\cal T}_{\{vw_1, vw_2\}}[R_0, R_1, R_2]\bigr).
        \end{split}
    \]

    \textit{Invariants.}
    Now, we show that for each partition $(R_0, R_1, R_2)$ of $X_t$, $\hat {\cal S}_t[R_0, R_1, R_2]$ maintains the correctness invariant.
    Consider any optimal solution $L\in \mathscr{L}$ and let $L_t = L \cap E_t$, $L_R = L\setminus L_t$.
    Let $R_i\subseteq X_t$ ($i = 0, 1, 2$) be the vertex set containing degree-$i$ vertices in $G_t[L_t]$. 
    Then we have that $L_t \in \mathcal{S}_t[R_0,R_1,R_2]$ and $L_t \cup L_R \in \mathscr{L}$.
    To prove the correctness invariants of $\hat{\mathcal{S}}_t[R_0,R_1,R_2] $, it is sufficient to prove that there exists a partial solution $\hat L_t\in \hat {\cal S}_t[R_0, R_1, R_2]$ such that $\hat L_t\cup L_R\in \mathscr{L}$ and $\hat L_t \in \mathcal{S}_t[R_0,R_1,R_2]$.

        We first show how to construct $\hat L_t$.
        We let $ L_{t'} = L \cap E_{t'}$, $L_R' = L\setminus L_{t'}$. 
        Without loss of generality, we assume that $L_{t'} \in  \mathcal{S}_{t'}[R_0',R_1',R_2']$. We let $D = L_{t} \setminus L_t'$ and thus $L_{R} = L_{R'} \setminus D$.
        By the correctness invariants of $\hat {\mathcal{S}}_{t'}[R_0',R_1',R_2']$, there exists $\hat L_{t'} \in \hat {\mathcal{S}}_{t'}[R_0',R_1',R_2']$ such that $\hat L_{t'} \cup L_{R}' \in \mathscr{L}$ and $\hat L_{t'} \in {\mathcal{S}}_{t'}[R_0',R_1',R_2']$.
        We let $\hat L_t = \hat L_{t'} \cup D$.
        Since $\hat L_{t'} \cup L_{R}' = \hat L_{t'} \cup (L_{R} \cup D) \in \mathscr{L}$, we have that $(\hat {L} _{t'} \cup D)\cup L_{R} = \hat{L}_t \cup L_R \in \mathscr{L}$.

        Note that $L_{t'} \in  \mathcal{S}_{t'}[R_0',R_1',R_2']$ and $L_t = L_{t'} \cup D \in  \mathcal{S}_{t}[R_0,R_1,R_2]$. It is not hard to verify that $\hat{L_t} = \hat {L}_t' \cup D \in \mathcal{S}_t[R_0,R_1,R_2]$ and $\hat{L_t} \in \mathcal{T}_{D}[R_0,R_1,R_2]$. Then our algorithm will put $\hat {L}_{t'}\cup D$ into $\hat {\mathcal{S}}_{t}[R_0,R_1,R_2]$.
        
        Therefore, there exists $\hat L_{t} \in \hat {\mathcal{S}}_{t}[R_0,R_1,R_2]$ such that $\hat L_{t} \cup L_{R} \in \mathscr{L}$ and $\hat L_{t} \in {\mathcal{S}}_{t}[R_0,R_1,R_2]$.
        By Lemma~\ref{Lemma:shrinking-cPST}, we know that $\hat{S}_{t}[R_0,R_1,R_2]$ still maintains the correctness invariants after we compute representative sets for each $(R_0,R_1,R_2)$.

    \textit{Running Time.}
    
    For each partition $(R_0', R_1', R_2')$ of $X_t$, it holds that $|\hat {\cal S}_{t'}[R_0', R_1', R_2']| \leq 2^{|R_0'| + |R_1'|}$ by the size invariant.
    Consider a dynamic programming transformation from $\hat {\cal S}_{t'}[R_0', R_1', R_2']$ to $\hat {\cal S}_{t}[R_0, R_1, R_2]$.
    It is not hard to check that $|R_0'| + |R_1'|\leq 3 + |R_0| + |R_1|$.
    
    Thus, we have that
    \[
        \begin{split}
                &|{\cal T}_{\emptyset}[R_0, R_1, R_2]| = O(2^{|R_0| + |R_1|})\\
                &|\bigcup_{vw\in E_v[X_t]}{\cal T}_{\{vw\}}[R_0, R_1, R_2]| = O(\mathbf{tw}\cdot 2^{|R_0| + |R_1|})\\
                &|\bigcup_{vw_1, vw_2\in E_v[X_t]}{\cal T}_{\{vw_1, vw_2\}}[R_0, R_1, R_2]| = O(\mathbf{tw}^2\cdot 2^{|R_0| + |R_1|}).
        \end{split}
    \]
    Thus, the size of $\hat {\cal S}_t[R_0, R_1, R_2]$ is bounded by  $O(\mathbf{tw}^2\cdot 2^{|R_0| + |R_1|})$.
    Then, we apply Lemma \ref{Lemma:shrinking} to reduce the size to $2^{|R_0| + |R_1|}$. 
    Similarly, the running time of computing all possible partition $(R_0, R_1, R_2)$ of $X_t$ is bounded by
    
    \[
        O\left(\sum_{i = 1}^{\mathbf{tw} + 1} {\mathbf{tw} + 1\choose i} 2^{i} (2^{i(\omega - 1)}i^{O(1)}(\mathbf{tw}^2\cdot 2^{i})n)\right) = O((1 + 2^{\omega + 1})^{\mathbf{tw}}\mathbf{tw}^{O(1)}n).
    \]

\vspace{2mm}
\noindent
    \textbf{Join Node $t$ with Two Children $t_1$ and $t_2$.} By definition, we have that $X_t = X_{t_1} = X_{t_2}$. 
    and $G_t = (V(G_{t_1})\cup V(G_{t_2}), E(G_{t_1})\cup E(G_{t_2}))$. 
    For each partition $(R_0, R_1, R_2)$ of $X_t$, we first give the definition of \emph{compatible partition pair}.
    \begin{definition}
        Given a partition $(R_0, R_1, R_2)$ of $X_t$, we call two partitions $(R_0', R_1', R_2')$ and $(R_0'', R_1'', R_2'')$ are \emph{compatible} with respect to $(R_0, R_1, R_2)$ if
        \begin{enumerate}
            \item $R_0\subseteq R_0'$ and $R_0\subseteq R_0''$,
            \item for each $v\in R_1$, either $v\in R_1' \wedge v\in R_0''$ or $v\in R_0' \wedge v\in R_1''$, and 
            \item for each $v\in R_2$, either $v\in R_2' \wedge v\in R_0''$,  $v\in R_1' \wedge v\in R_1''$, or $v\in R_0' \wedge v\in R_2''$.
        \end{enumerate} 
        
    \end{definition}

    For every partition $(R_0, R_1, R_2)$ of $X_t$, we will take the union of edge subsets of the families stored at $\hat {\cal S}_{t_1}[R_0', R_1', R_2']$ and $\hat {\cal S}_{t_2}[R_0'', R_1'', R_2'']$ where $(R_0', R_1', R_2')$ and $(R_0'', R_1'', R_2'')$ are compatible with respect to $(R_0, R_1, R_2)$.
    That is 
    \begin{equation*}
    \begin{aligned}      
        & \hat {\cal S}_t[R_0, R_1, R_2] := & \\
        & \bigcup_{\substack{(R_0', R_1', R_2') \text{ and } (R_0'', R_1'', R_2'')\\ \text{ are compatible}}} \{E_1 \cup E_2: &E_1\in \hat {\cal S}_{t_1}[R_0', R_1', R_2'] \wedge E_2\in \hat {\cal S}_{t_2}[R_0'', R_1'', R_2'']\wedge\\
        & &G_t[E_1\cup E_2] \text{ is a collection of induced paths}\}.
    \end{aligned}
    \end{equation*}
    
    Finally, for each partition $(R_0, R_1, R_2)$ of $X_t$, if $|\hat {\cal S}_t[R_0, R_1, R_2]|\geq 2^{|R_0| + |R_1|}$, we apply Lemma \ref{Lemma:shrinking} to reduce the size to $2^{|R_0| + |R_1|}$.

    \textit{Invariants.}
    Now, we show that for each partition $(R_0, R_1, R_2)$ of $X_t$, $\hat {\cal S}_t[R_0, R_1, R_2]$ maintains the correctness invariant.
    Consider any optimal solution $L\in \mathscr{L}$ and let $L_t = L \cap E_t$, $L_R = L\setminus L_t$.
    Let $R_i\subseteq X_t$ ($i = 0, 1, 2$) be the vertex set containing degree-$i$ vertices in $G_t[L_t]$. 
    We further let $L_{t_1} = L\cap E_{t_1}$, $L_{t_2} = L\cap E_{t_2}$. By the definition of tree decomposition, we have that $L_t = L_{t_1} \cup L_{t_2}$.
    Let $R_i'\subseteq X_{t_1}$ (resp. $R_i''\subseteq X_{t_2}$) ($i = 0, 1, 2$) be the vertex set containing degree-$i$ vertices in $G_{t_1}[L_{t_1}]$ (resp. $G_{t_2}[L_{t_2}]$). 
    By the definition of compatible partition pairs, we know that $(R_0', R_1', R_2')$ and $(R_0'', R_1'', R_2'')$ are compatible.
    Similar to the argument in the forget node, the goal is to show that there exists a partial solution $\hat L_t\in \hat {\cal S}_t[R_0, R_1, R_2]$ such that $\hat L_t\cup L_R\in \mathscr{L}$ and $R_i (i = 0, 1, 2)$ exactly contains degree-$i$ vertices in $G_t[\hat L_t]$.

    By the correctness invariant of $\hat {\cal S}_{t_1}[R_0', R_1', R_2']$, there exists a partial solution $\hat L_{t_1}\in \hat {\cal S}_{t_1}[R_0', R_1', R_2']$ such that $\hat L_{t_1}\cup (L_{t_2}\cup L_R)\in \mathscr{L}$ such that $R_i' (i = 0, 1, 2)$ exactly contains degree-$i$ vertices in $G_{t_1}[\hat L_{t_1}]$.
        
    Similarly, by the correctness invariant of $\hat {\cal S}_{t_2}[R_0'', R_1'', R_2'']$, there exists a partial solution $\hat L_{t_2}\in \hat {\cal S}_{t_2}[R_0'', R_1'', R_2'']$ such that $\hat L_{t_2}\cup (\hat L_{t_1}\cup L_R)\in \mathscr{L}$ and $R_i'' (i = 0, 1, 2)$ exactly contains degree-$i$ vertices in $G_{t_2}[\hat L_{t_2}]$.
        
    Let $\hat L_t = \hat L_{t_1}\cup \hat L_{t_2}$ and our algorithm indeed puts $\hat L_t$ into $\hat {\cal S}_t[R_0, R_1, R_2]$.
    Since $\hat L_{t_1}\cup \hat L_{t_2} \cup L_R\in \mathscr{L}$, we have that $\hat L_t\cup L_R\in \mathscr{L}$.
    
    Thus, the correctness invariant is maintained.   
    To maintain the size invariant, we just apply Lemma \ref{Lemma:shrinking-cPST} for each $\hat {\cal S}_t(R_0, R_1, R_2)$.

    \textit{Running Time.}
    By the definition of compatible partition pairs, we have that that the number of compatible partition pairs is bounded by $2^{|R_1|} \cdot 3^{|R_2|}$.
    For each compatible partition pair $(R_0', R_1', R_2')$ and $(R_0'', R_1'', R_2'')$, we know that the size of $\hat {\cal S}_{t_1}[R_0', R_1', R_2']$ is at most $2^{|R_0'| + |R_1'|}$ by the size invariant.

    For simplicity of the complexity analysis, we may employ the coarser upper bound $|\hat{\mathcal{S}}_{t_1}[R_0', R_1', R_2']| \leq 2^{|R_0'| + |R_1'| + |R_2'|}$. This bound suffices to ensure a single exponential time algorithm.
    Similarly, we have that $|\hat {\cal S}_{t_2}[R_0'', R_1'', R_2'']|\leq 2^{|R_0''| + |R_1''| + |R_2''|}$.
    
    The size of $\hat {\cal S}_t[R_0, R_1, R_2]$ is bounded by 
    \[
        2^{|R_1|} \cdot 3^{|R_2|} \cdot 2^{|R_0'| + |R_1'| + |R_2'|}\cdot 2^{|R_0''| + |R_1''| + |R_2''|} = 2^{|R_1|}\cdot  3^{|R_2|} \cdot  4^{|R_0| + |R_1| + |R_2|}.
    \]
    
    Now, we consider the running time of applying Lemma \ref{Lemma:shrinking} for all possible partition $(R_0, R_1, R_2)$ of $X_t$. which is bounded by 
    
    \[
        \begin{split}
            & O\left(\sum_{i = 1}^{\mathbf{tw} + 1} {\mathbf{tw} + 1\choose i} \sum_{j = 1}^{i} {i\choose j}  4^{\mathbf{tw}} 3^{\mathbf{tw}-i} 2^{i - j} (2^{i(\omega - 1)}\mathbf{tw}^{O(1)}n)\right) \\
            = & O\left(12^{\mathbf{tw}}\sum_{i = 1}^{\mathbf{tw} + 1} {\mathbf{tw} + 1\choose i} (\frac{1}{3} \cdot 2^{\omega})^i  \sum_{j = 1}^{i} {i\choose j} (1/2)^j  \mathbf{tw}^{O(1)}n\right) \\
            = & O\left(12^{\mathbf{tw}}\sum_{i = 1}^{\mathbf{tw} + 1} {\mathbf{tw} + 1\choose i} (\frac{1}{3} \cdot 2^{\omega})^i (3/2)^i \mathbf{tw}^{O(1)}n\right) \\
            = & O\left(12^{\mathbf{tw}}\sum_{i = 1}^{\mathbf{tw} + 1} {\mathbf{tw} + 1\choose i} (2^{\omega - 1})^i \mathbf{tw}^{O(1)}n\right) \\
            = & O((12 + 3 \cdot 2^{\omega + 1})^{\mathbf{tw}} \mathbf{tw}^{O(1)}n).
        \end{split}
    \]

    Thus, the whole running time of a join node is bounded by $O((12 + 3 \cdot 2^{\omega + 1})^{\mathbf{tw}} \mathbf{tw}^{O(1)}n)$.

    \vspace{2mm}

    Since the number of node in a nice tree decomposition is bounded by $O(n)$, the whole algorithm takes $O((12 + 3 \cdot 2^{\omega + 1})^{\mathbf{tw}} \mathbf{tw}^{O(1)}n^2)$ time.
    However, it is not necessary to compute the connectivity graphs and the weights corresponding to the partial solutions in each node.
    Since the size of a connectivity graph is at most $\mathbf{tw} + 1$, we can maintain all connectivity graphs with weights during the dynamic programming process in time $\mathbf{tw}^{O(1)}$.
    To get linear time algorithms, we only store connectivity graphs and update them instead of storing partial solutions.
    For join node, it is possible to judge the combined partition solution is feasible only by combining the connectivity graphs.
    
    Thus, we can reduce the running time to $O((12 + 3 \cdot 2^{\omega + 1})^{\mathbf{tw}} \mathbf{tw}^{O(1)}n)$.
    
    For pathwidth, since there is no join node in a nice path decomposition, our algorithm runs in $O((1 + 2^{\omega + 1})^{\mathbf{pw}}\mathbf{pw}^{O(1)}n)$ time.
    This theorem holds.
\end{proof}

\section{The algorithm for \textsc{Co-Path Packing}}
In this section, we consider an optimal and weighted version of \textsc{Co-Path Packing}, called \textsc{Weighted Co-Path Packing}. 
In this problem, we seek a maximum weight subgraph that is a collection of induced paths.
The formal definition is as follows.

\begin{tcolorbox}
\textsc{Weighted Co-Path Packing}\\
\textbf{Instance:} A graph $G=(V, E)$ and a weight function $w : V(G) \to \mathbb{N}$.\\
\textbf{Task:} Find a subset of $V(G)$ of maximum weight whose induced graph is a collection of induced paths.\\
\vspace{-5mm}
\end{tcolorbox}

In this section, 
a vertex subset $V'\subseteq V(G)$ is called a \emph{solution} if $G[V']$ is a collection of induced paths.
A vertex subset $V'\subseteq V(G)$ is called an \emph{optimal solution} if $V'$ is a solution of the maximum weight.
Let $\mathscr{L}$ be a family of vertex subsets that contains all optimal solutions.

Assuming there is a tree decomposition ${\cal T} = (T, \{X_t\}_{t\in V(T)})$ of $G$.
For any node $t$, we use $V_t$ to denote the vertex set $\bigcup_{t'} X_{t'}$, where $t'$ ranges over $t$ and all descendants of $t$. Let $G_t = (V_t, E(G[V_t]) \setminus E(X_t))$. 
We further use ${\cal S}_t$ to denote the family containing all vertex subsets of $V_t$ that induce a collection of induced paths in $G_t$.
We call ${\cal S}_t$ a \emph{family of partial solutions} of $t$. 
Then, for each partition $(D, R_0, R_1, R_2)$ of $X_t$, we define a family of vertices ${\cal S}_t[D, R_0, R_1, R_2]$ that contains partial solutions of $t$ as follows.

\begin{definition}
    Consider a node $t \in T$ and a partition $(D, R_0, R_1, R_2)$ of $X_t$. 
    We define $\mathcal{S}_t[D, R_0, R_1, R_2]$ as a family of all partial solutions $L_t \subseteq V_t$ such that
    \begin{enumerate}
        
        \item $D := X_t\setminus L_t$, and 
        
        \item $R_i\subseteq X_t \cap L_t$ ($i = 0, 1, 2$) is a subset of $X_t$ containing degree-$i$ vertices in $G_t[L_t]$, respectively.
    \end{enumerate}
    \label{def-s_t}
\end{definition}

Clearly, $\mathcal{S}_t = \bigcup_{(D, R_0, R_1, R_2) \text{ is a partition of } X_{t}}{\cal S}_{t}[D, R_0, R_1, R_2]$.
In our approach, 
rather than maintaining a standard DP table for each node $t$ of the tree decomposition,
We keep a family of partial solutions for the graph $G_t$. 
This family of partial solutions is constructed such that for every optimal solution $L\in \mathscr{L}$ and its corresponding intersection $L_t = V_t \cap L$ with the graph $G_t$, there exists a partial solution $\hat L_t$ in the family such that $\hat L_t \cup (L\setminus L_t)$ is also an optimal solution. 
Specifically, for each node $t$ and each partition $(D, R_0, R_1, R_2)$ of $X_t$, we keep a family of vertices $\hat {\cal S}_t[D, R_0, R_1, R_2]$ of $G_t$ satisfying the following correctness invariant.

\vspace{2mm}
\noindent
\textbf{Correctness Invariant.} For each $L\in \mathscr{L}$, let $L_t = L \cap V_t$ and $L_R = L\setminus L_t$.
Let $(D, R_0, R_1, R_2)$ be the partition of $X_t$ such that $L_t\in \mathcal{S}_t[D, R_0, R_1, R_2]$.
Then, there exists $\hat L_t\in \hat {\cal S}_t[D, R_0, R_1, R_2]$ such that (i) $\hat L = \hat L_t \cup L_R \in \mathscr{L}$ and (ii) $\hat L_t\in {\cal S}_t[D, R_0, R_1, R_2]$.

\vspace{2mm}

Our dynamic programming algorithm traverses the nodes of the tree $T$ from the leaf nodes to the root node.
To show the correctness of the algorithm, we maintain this correctness invariant throughout this dynamic programming.
Then, we show that we can use representative family to obtain $\hat {\cal S}_t[D, R_0, R_1, R_2]$ of small size.
Specifically, we maintain the following size invariant:

\vspace{2mm}
\noindent
\textbf{Size Invariant.} For each node $t$ and each partition $(D, R_0, R_1, R_2)$, after the process of the algorithm, we have that $|\hat {\cal S}_t[D, R_0, R_1, R_2]|\leq 2^{|R_0| + |R_1|}$.
\vspace{2mm}

The difficulty in reducing ${\cal S}_t[D, R_0, R_1, R_2]$ to a single exponential size is that we need to consider all possible connectivity situations for the vertices in $R_0 \cup R_1$, of which there are $O((|R_0| + |R_1|)^{|R_0| + |R_1|})$.
However, by using the representative family technique,
we can only consider a single exponential size of connectivity situations for $R_0 \cup R_1$.
Firstly, we define the concept of \emph{connectivity graphs} to denote the connectivity situations of $R_0\cup R_1$.
Though vertices in $R_0$ are isolated vertices in the partial solution, they might be adjacent to some vertices not in $V_t$ in the optimal solution.
Thus, we should also consider $R_0$ in the connectivity graph.

For some vertex set $V'\in V(G)$, let $w(V') = \sum_{v\in V'} w(v)$.

\begin{definition}[Connectivity Graph]\label{def-con-graph}    
    Let $t$ be a node of $T$ and $(D, R_0, R_1, R_2)$ be a partition of $X_t$.
    let $E'\subseteq E(G)$ be an edge set such that $G[E']$ is a collection of induced paths and $R_i \subseteq X_t$ ($i\in\{1,2\}$) is the set of vertices of degree $i$ in $G[E']$.
    The \emph{connectivity graph} $F_t(E') = (V_F, E_F)$ with respect to $(D, R_0, R_1, R_2)$ is obtained as follows:
    Each vertex in $V_F$ corresponds to a vertex in the vertex set $R_0\cup R_1$. 
    Two vertices $u, v\in V_F$ are adjacent if and only if the vertices in $R_1$ corresponding to $u$ and $v$ belong to the same induced path in $G[E']$.

\end{definition}

\begin{figure}[!t]
    \centering
    \includegraphics[height = 6cm]{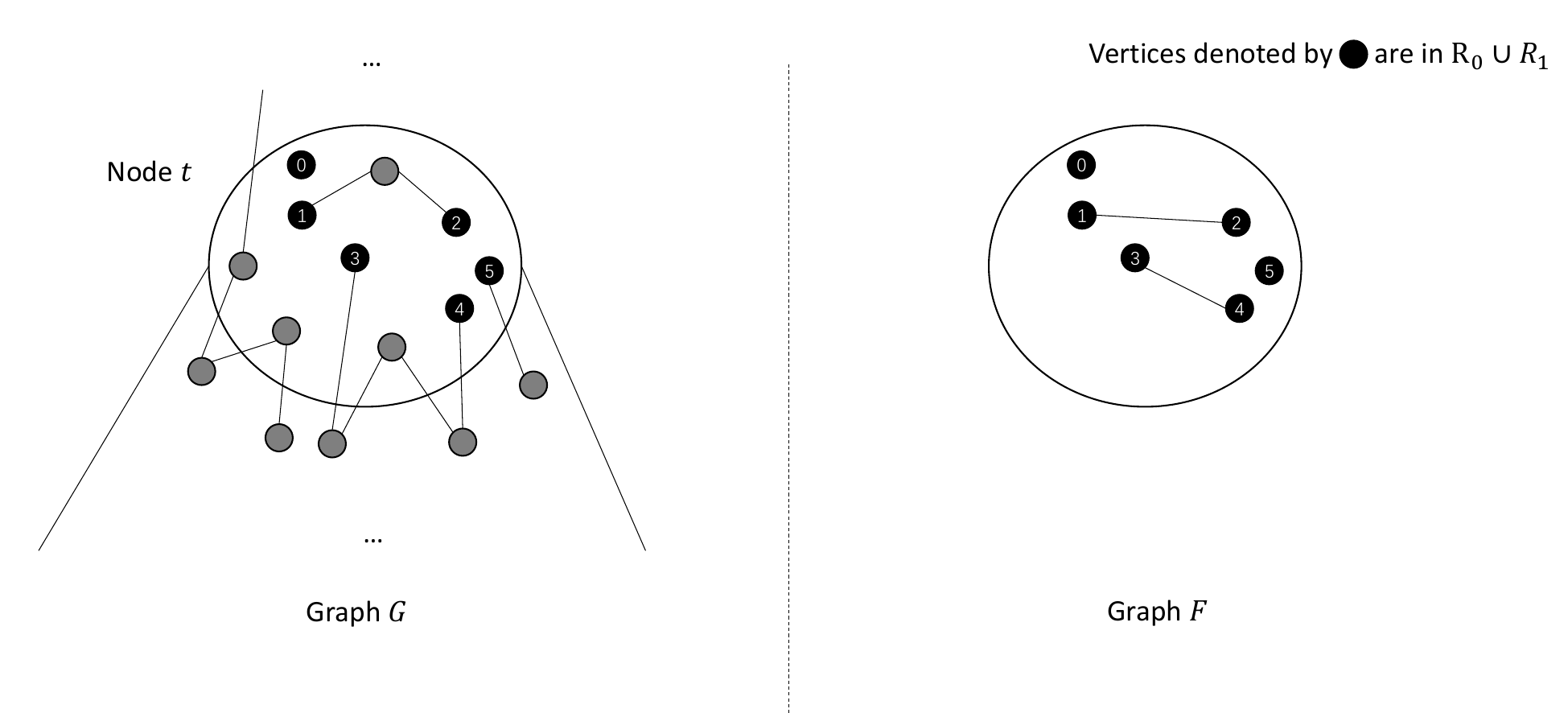}
    \caption{An example for a connectivity graph of $t$ and $E'$ with respect to $(D, R_0, R_1,R_2)$.
    The vertices marked black are in $V(F)$. The vertex marked 0 is in $R_0$ and the vertices marked 1-5 are in $R_1$. The vertices marked 1 and 2 (3 and 4) are adjacent since they are in the same induced path in $G[E']$.}
    \label{fig:ConnectivityGraph}
\end{figure}

Connectivity graphs are used to track which pairs of vertices in $R_1$ belong to the same induced path of $G[E']$.

See Fig. \ref{fig:ConnectivityGraph} for an illustration.
We use $F(E')$ to denote $F_t(E')$ with respect to $(D, R_0, R_1, R_2)$ when node $t$ and partition $(D, R_0,R_1,R_2)$ is clear in context.
To show how to maintain the size invariant, we have the following lemma.

\begin{lemma}\label{Lemma:shrinking}
    Let $t$ be a node of $T$ and $(D, R_0, R_1, R_2)$ be a partition of $X_t$. 
    Let $\hat {\cal S}_t[D, R_0, R_1, R_2]$ be a family of vertex subsets of $G_t$ satisfying the correctness invariant with $|\hat {\cal S}_t[D, R_0, R_1, R_2]| = l$ and $|R_0| + |R_1| = k$.
    Then, we can compute $\hat {\cal S}'_t[D, R_0, R_1, R_2]\subseteq \hat {\cal S}_t[D, R_0, R_1, R_2]$ satisfying the correctness invariant and size invariant in $O(2^{k(\omega - 1)}k^{O(1)}nl)$ time.

\end{lemma}

\begin{proof}

    Let $K[R_0\cup R_1]$ be a complete graph on the vertex set $R_0\cup R_1$.
    We consider a graphic matroid $M$ on $K[R_0\cup R_1]$.

    Let $\hat {\cal S}_t[D, R_0, R_1, R_2] = \{V_1^t, \dots, V_l^t\}$ and let $E_i^t = E(G_t[V_i^t])$ for each $1\leq i\leq l$.
    Let ${\cal N} = \{F(E_1^t), \dots, F(E_l^t)\}$. We also define $w(F_t(E_i^t)) = w(V_i^t)$ for each $1\leq i\leq l$.
    For each $i\in \{1,\dots, k - 1\}$, let ${\cal N}_i$ be a subset of ${\cal N}$ containing all path collections with $i$ edges.
    For each ${\cal N}_i$ we apply Lemma \ref{Lemma:RepComputation} to compute its max $(k - 1 - i)$-representative family
    \[
        \hat {\cal N}_i\subseteq_{maxrep}^{k - 1 - i} {\cal N}_i.
    \]
    We now construct $\hat {\cal S}'_t[D, R_0, R_1, R_2]\subseteq \hat {\cal S}_t[D, R_0, R_1, R_2]$ as follows. We put every $V_j^t\in \hat {\cal S}_t[D, R_0, R_1, R_2]$ into $\hat {\cal S}'_t[D, R_0, R_1, R_2]$ if $F(E_j^t) \in \bigcup_{i = 1}^{k - 1}\hat{\cal N}_i$.
    By Lemma \ref{Lemma:RepComputation}, we have that $|\hat {\cal S}'_t[D, R_0, R_1, R_2]|\leq \sum_{i = 1}^{k - 1}{k - 1\choose i}\leq 2^{k} = 2^{|R_0|+|R_1|}$. Thus, the size invariant holds $\hat {\cal S}'_t[D, R_0, R_1, R_2]$.
    
    Next, we show that the correctness invariant holds.
    Consider any optimal solution $L\in \mathscr{L}$ such that $D := X_t\setminus L$ and $R_i\subseteq X_t$ ($i = 0, 1, 2$) be the vertex set containing degree-$i$ vertices in $G_t[L]$.
    
    Let $L_t = L \cap V_t$, $L_R = L\setminus L_t$.
    We further let $E_R = E(G[L])\setminus E(G_t[L_t])$. 
    Since $\hat {\cal S}_t[D, R_0, R_1, R_2]$ satisfies the correctness invariant,
    we have that there exists a partial solution $V_j^t\in \hat {\cal S}_t[D, R_0, R_1, R_2]$ such that $\hat L = V_j^t \cup L_R$ is an optimal solution and $D = X_t\setminus \hat L$ and $R_i\subseteq X_t$ ($i = 0, 1, 2$) exactly contains degree-$i$ vertices in $G_t[\hat L_t]$.
    Since $\hat L$ is a feasible solution, we have that $F_t(E_j^t\cup E_R)$ is a collection of induced paths.
    
    Suppose that $F_t(E_j^t) \in \mathcal{N}_i$ for some $i\geq 0$.
    Since $\hat {\cal N}_i\subseteq_{maxrep}^{k - 1 - i} {\cal N}_i$, we have that there exists $F(E_h^t)\in \hat {\cal N}_i$ such that $w(F(E_h^t))\geq w(F(E_j^t))$ and there is no cycle in the graph $F_t(E_h^t\cup E_R)$. 
    Let $\hat L' = V_h^t\cup L_R$. 

    We then prove that $\hat L'$ is an optimal solution.
    First, we have that $w(\hat L') = w(L_R) + w(F(E_h^t)) \geq w(L_R) + w(F(E_j^t)) = w(\hat L)$.
    Since each degree of vertices in $V_j^t\cap X_t$ is the same as the corresponding vertices in $V_h^t\cap X_t$, there is no vertex of degree at least 3 in the graph $G[\hat L']$. 
    Recall that $G[E_h^t]$ and $G[E_R]$ are both collections of induced paths.    
    Suppose that there exists a cycle in the graph $G[\hat L']$.
        This cycle must contain at least one induced path $P_1$ in $G[E_R]$ and at least one induced path $P_2$ in $G[E_h^t]$.
        Let $v_1$ and $v_2$ be one of the two endpoints of $P_1$ and $P_2$, respectively.
        Then $v_1$ and $v_2$ are both in $R_1$ and thus in $F_t(E_h^t\cup E_R)$.
    However, in this case,  $v_1$ and $v_2$ are contained in a cycle in the graph $F_t(E_h^t\cup E_R)$, which is a contradiction.
    Thus the graph $G[\hat L']$ does not contain any cycles.
    Therefore, $\hat L'$ is an optimal solution.

    By Lemma~\ref{Lemma:RepComputation}, the running time to compute $\hat {\cal S}'_t[D, R_0, R_1, R_2]$ is bounded by
    \[
        O\left(l\left(\sum_{i = 1}^{k - 1}{k - 1\choose i}k^\omega + \sum_{i = 1}^{k - 1}{k - 1\choose i}^{\omega - 1}\right)\right),
    \]
    which is bounded by
    \[
        O\left(l\left(\sum_{i = 1}^{k - 1}{k - 1\choose i}^{\omega - 1}k^\omega\right)\right)=O(2^{k(\omega - 1)}k^{O(1)}l).
    \]
    For a given vertex set, we need to compute the corresponding connectivity graph and that takes $O(n)$ time.
    Thus, this lemma holds.
\end{proof}
 
In the dynamic programming process, the size of $\hat {\cal S}_t[D, R_0, R_1, R_2]$ can be larger than $2^{|R_0| + |R_1|}$.
Then we can use Lemma \ref{Lemma:shrinking} to reduce its size to at most $2^{|R_0| + |R_1|}$.
Now, we are ready to introduce our main algorithm in this section.

\vspace{2mm}
\noindent\textbf{Theorem 2.}
\textit{
    Given a tree decomposition of $G$ with width $\mathbf{tw}$. \textsc{Co-Path Packing} can be solved in $O((13 + 3\cdot 2^{\omega + 1})^{\mathbf{tw}} \mathbf{tw}^{O(1)}n)$ time.
    Given a path decomposition of $G$ with width $\mathbf{pw}$. \textsc{Co-Path Packing} can be solved in $O((2 + 2^{\omega + 1})^{\mathbf{pw}}\mathbf{pw}^{O(1)}n)$ time.
}

\begin{proof}
    We can simply assume that the tree decomposition ${\cal T} = (T, \{X_t\}_{t\in V(T)})$ is a nice tree decomposition by Lemma~\ref{Lemma:nice}.
    We show how $\hat {\cal S}_t[D, R_0, R_1, R_2]$ is obtained for each node $t$ and each partition $(D, R_0, R_1, R_2)$ of $X_t$ from leaf nodes to the root node.
    For each node type, we first show the process of our algorithm and then show the maintenance of two invariants and time complexity.
    
\vspace{2mm}
\noindent
    \textbf{Leaf Node $t$.} It trivially holds that $\hat {\cal S}_t[D, R_0, R_1, R_2] = \emptyset$ for each node $t$ and each partition $(D, R_0, R_1, R_2)$ of $X_t$.
    Clearly, the correctness and size invariants are maintained.

\vspace{2mm}
\noindent
    \textbf{Introduce Node $t$ with Child $t'$.} Let $X_t = X_{t'}\cup \{v\}$ for some $v\notin X_{t'}$.
    By the definition of tree decomposition, we know that all neighbors of $v$ in the graph $G_t$ must be in $X_{t'}$.
    Thus, by the definition of $G_t$, we know that $v$ is a degree-0 vertex in $G_t$.
    For every partition $(D, R_0, R_1, R_2)$ of $X_t$, one of the following three cases will happen.
    \begin{enumerate}
        \item $v\in R_1\cup R_2$. It holds that $\hat {\cal S}_t[D, R_0, R_1, R_2] =\emptyset$ since $v$ is a degree-0 vertex in $G_t$.
        
        \item $v\in R_0$. It holds that 
        \[
            \hat {\cal S}_t[D, R_0, R_1, R_2] = \{V'\cup \{v\}\mid V'\in \hat {\cal S}_{t'}[D, R_0\setminus \{v\}, R_1, R_2]\}.
        \]
        
        \item $v\in D$. It holds that $\hat {\cal S}_t[D, R_0, R_1, R_2] =\hat {\cal S}_{t'}[D\setminus \{v\}, R_0, R_1, R_2]$.
        
    \end{enumerate}
    Since the graph $G_t$ only differs from $G_{t'}$ at $v$, we indeed consider two cases whether $v$ is in partial solutions or not. 
    Note that our algorithm does not discard any partial solutions in this process.
    It is not hard to see that the correctness and size invariants are maintained.

\vspace{2mm}
\noindent
    \textbf{Forget Node $t$ with Child $t'$.} Let $X_t = X_{t'}\setminus \{v\}$ for some $v\notin X_{t}$.
    Let $E_v[X_t]$ denote the set of edges between $v$ and the vertices in $X_t$. 
    It is easy to see that $V(G_t) = V(G_{t'})$ and $E(G_t) = E(G_{t'})\cup E_v[X_t]$.

    Consider a partition $(D, R_0, R_1, R_2)$ of $X_t$.
    We consider whether $v$ is in partial solutions or not. 
    If $v$ is not in partial solutions, we know the corresponding partial solution set is  
    $\hat {\cal S}_{t'}[D\cup \{v\}, R_0, R_1, R_2]$.
    If $v$ is in partial solutions, then $v$ and its neighbors in $X_t$ must satisfy the degree constraints.
    We give the definition of \emph{feasible partitions} of $X_{t'}$ for the valid degree cases.
    \begin{definition}
        Consider a partition $(D, R_0, R_1, R_2)$ of $X_t$. Let $W = N_{G_t}(v) \cap (R_1\cup R_2)$.
        we call a partition $(D', R_0', R_1', R_2')$ of $X_{t'}$ is \emph{feasible} with respect to $(D, R_0, R_1, R_2)$ if $D\setminus (W\cup \{v\}) = D' \setminus (W\cup \{v\})$, $R_i \setminus (W\cup \{v\}) = R_i' \setminus (W\cup \{v\})$ for $i \in \{0, 1, 2\}$ and one of the following conditions holds.
        \begin{enumerate}
            \item $|W| = 0$;
            \item $|W| = 1$ (where $W = \{w_1\}$) and there exists $i, j \in \{0, 1\}$ such that $w_1\in R_i' \cap R_{i + 1}$ and $v\in R_j' \cap R_{j + 1}$;
            \item $|W| = 2$ (where $W = \{w_1, w_2\}$), $v\in R_0'\cap R_2$, and there exists $i, j\in \{0, 1\}$ such that $w_1\in R_i'\cap R_{i + 1}$ and $w_2\in R_j'\cap R_{j + 1}$.
        \end{enumerate} 
        
    \end{definition}

    Thus, it holds that 
    \begin{equation*}
    \begin{aligned}      
        & \hat {\cal S}_t[D, R_0, R_1, R_2] := \hat {\cal S}_{t'}[D\cup \{v\}, R_0, R_1, R_2] \cup\\
        & \bigcup_{\substack{(D', R_0', R_1', R_2') \\ \text{ is feasible}}} \{V': V'\in \hat {\cal S}_{t'}[D', R_0', R_1', R_2'] \wedge G_t[V'] \text{ is a collection of induced paths}\}.
    \end{aligned}
    \end{equation*}

    \textit{Invariants.}
    Since $V(G_t)$ is the same as $V(G_{t'})$ and $E(G_{t'})\subseteq E(G_{t})$, the set of valid partial solutions in $G_t$ is a subset of those in $G_{t'}$.
    In this step, the algorithm specifically prunes partial solutions that violate the acyclicity constraint (i.e., those that form a cycle).      
    Thus the correctness invariant is maintained.

    To maintain the size invariant, we just apply Lemma \ref{Lemma:shrinking} for each $\hat {\cal S}_t(D, R_0, R_1, R_2)$.

    \textit{Running Time.}
    By the definition of feasible partitions, the number of feasible partitions is at most 4.
    Thus, it is easy to see that $|\hat {\cal S}_t[D, R_0, R_1, R_2]| \leq 2^{|R_0| + |R_1|} + 4\cdot 2^{|R_0| + |R_1| + 2} = 17\cdot 2^{|R_0| + |R_1|}$ and then we apply Lemma \ref{Lemma:shrinking} to reduce the size to $2^{|R_0| + |R_1|}$. 
    The running time of computing all possible partition $(D, R_0, R_1, R_2)$ of $X_t$ is bounded by 
    \[
        O\left(\sum_{i = 1}^{\mathbf{tw} + 1} {\mathbf{tw} + 1\choose i} 2^{i}2^{\mathbf{tw} + 1 - i} (2^{i(\omega - 1)}i^{O(1)}(17\cdot 2^{i})n)\right) = O((2 + 2^{\omega + 1})^{\mathbf{tw}}\mathbf{tw}^{O(1)}n).
    \]

\vspace{2mm}
\noindent
    \textbf{Join Node $t$ with Two Children $t_1$ and $t_2$.} By definition, we have that $X_t = X_{t_1} = X_{t_2}$. 
    and $G_t = (V(G_{t_1})\cup V(G_{t_2}), E(G_{t_1})\cup E(G_{t_2}))$. 
    
    For each partition $(D, R_0, R_1, R_2)$ of $X_t$, we first give the definition of \emph{compatible partition pair}.
    \begin{definition}
        Given a partition $(D, R_0, R_1, R_2)$ of $X_t$, we call two partitions $(D', R_0', R_1', R_2')$ and $(D'', R_0'', R_1'', R_2'')$ are \emph{compatible} with respect to $(D, R_0, R_1, R_2)$ if
        \begin{enumerate}
            \item $D = D' = D''$, 
            \item $R_0\subseteq R_0'$ and $R_0\subseteq R_0''$, 
            \item for each $v\in R_1$, either $v\in R_1' \wedge v\in R_0''$ or $v\in R_0' \wedge v\in R_1''$, and 
            \item for each $v\in R_2$, either $v\in R_2' \wedge v\in R_0''$,  $v\in R_1' \wedge v\in R_1''$, or $v\in R_0' \wedge v\in R_2''$.
        \end{enumerate} 
        
    \end{definition}

    For every partition $(D, R_0, R_1, R_2)$ of $X_t$, we will take the union of vertex subsets of the families stored at $\hat {\cal S}_{t_1}[D', R_0', R_1', R_2']$ and $\hat {\cal S}_{t_2}[D'', R_0'', R_1'', R_2'']$ where $(D', R_0', R_1', R_2')$ and $(D'', R_0'', R_1'', R_2'')$ are compatible with respect to $(D, R_0, R_1, R_2)$.
    That is 
    \begin{equation*}
    \begin{aligned}   
        & \hat {\cal S}_t[D, R_0, R_1, R_2] := &\bigcup_{\substack{(D', R_0', R_1', R_2') \text{ and } (D'', R_0'', R_1'', R_2'')\\ \text{ are compatible}}} \{V_1 \cup V_2: ~~~~~~~~~~~~~~~~~ \\
        & & V_1\in \hat {\cal S}_{t_1}[D', R_0', R_1', R_2'] \wedge V_2\in \hat {\cal S}_{t_2}[D'', R_0'', R_1'', R_2'']\wedge\\
        & & G_t[V_1\cup V_2] \text{ is a collection of induced paths}\}.
    \end{aligned}
    \end{equation*}

    Finally, for each partition $(D, R_0, R_1, R_2)$ of $X_t$, if $|\hat {\cal S}_t[D, R_0, R_1, R_2]|\geq 2^{|R_0| + |R_1|}$, we apply Lemma \ref{Lemma:shrinking} to reduce the size to $2^{|R_0| + |R_1|}$.

    \textit{Invariants.}
    Now, we show that for each partition $(D, R_0, R_1, R_2)$ of $X_t$, $\hat {\cal S}_t[D, R_0, R_1, R_2]$ maintains the correctness invariant.
    Consider any optimal solution $L\in \mathscr{L}$ and let $L_t = L \cap V_t$, $L_R = L\setminus L_t$.
    Let $D := X_t\setminus L$ and $R_i\subseteq X_t$ ($i = 0, 1, 2$) be the vertex set containing degree-$i$ vertices in $G_t[L_t]$. 
    We further let $L_{t_1} = L\cap V_{t_1}$, $L_{t_2} = L\cap V_{t_2}$ and $L_X = L_{t_1}\cap L_{t_2} = L\cap X_t$. 
    Let $D' := X_{t_1}\setminus L$ (resp. $D'' := X_{t_2}\setminus L$) and $R_i'\subseteq X_{t_1}$ (resp. $R_i''\subseteq X_{t_2}$) ($i = 0, 1, 2$) be the vertex set containing degree-$i$ vertices in $G_{t_1}[L_{t_1}]$ (resp. $G_{t_2}[L_{t_2}]$). 
    By the definition of compatible partition pairs, we know that $(D', R_0', R_1', R_2')$ and $(D'', R_0'', R_1'', R_2'')$ are compatible.
        The goal is to show that there exists a partial solution $\hat L_t\in \hat {\cal S}_t[D, R_0, R_1, R_2]$ such that $\hat L_t\cup L_R\in \mathscr{L}$, $D = X_t\setminus \hat L_t$, and $R_i (i = 0, 1, 2)$ exactly contains degree-$i$ vertices in $G_t[\hat L_t]$.

        By the correctness invariant of $\hat {\cal S}_{t_1}[D', R_0', R_1', R_2']$, there exists a partial solution $\hat L_{t_1}\in \hat {\cal S}_{t_1}[D', R_0', R_1', R_2']$ such that $\hat L_{t_1}\cup ((L_{t_2}\setminus L_X)\cup L_R)\in \mathscr{L}$ such that $\hat L_{t_1}\cap X_t = L_X$, $D' = X_{t_1}\setminus \hat L_{t_1}$ and $R_i' (i = 0, 1, 2)$ exactly contains degree-$i$ vertices in $G_{t_1}[\hat L_{t_1}]$.
        
        Similarly, by the correctness invariant of $\hat {\cal S}_{t_2}[D'', R_0'', R_1'', R_2'']$, there exists a partial solution $\hat L_{t_2}\in \hat {\cal S}_{t_2}[D'', R_0'', R_1'', R_2'']$ such that $\hat L_{t_2}\cup ((\hat L_{t_1}\setminus L_X)\cup L_R)\in \mathscr{L}$, $D'' = X_{t_2}\setminus \hat L_{t_2}$ and $R_i'' (i = 0, 1, 2)$ exactly contains degree-$i$ vertices in $G_{t_2}[\hat L_{t_2}]$.
        
        Let $\hat L_t = \hat L_{t_1}\cup \hat L_{t_2}$ and our algorithm indeed puts $\hat L_t$ into $\hat {\cal S}_t[D, R_0, R_1, R_2]$.
        Since $\hat L_{t_1}\cup \hat L_{t_2} \cup L_R\in \mathscr{L}$, we have that $\hat L_t\cup L_R\in \mathscr{L}$.
        
        Thus, the correctness invariant is maintained.   
        To maintain the size invariant, we just apply Lemma \ref{Lemma:shrinking} for each $\hat {\cal S}_t(D, R_0, R_1, R_2)$.

    \textit{Running Time.}
    By the definition of compatible partition pairs, we have that that the number of compatible partition pairs is bounded by $2^{|R_1|} \cdot 3^{|R_2|}$.
    For each compatible partition pair $(D', R_0', R_1', R_2')$ and $(D'', R_0'', R_1'', R_2'')$, we know that the size of $\hat {\cal S}_{t_1}[D', R_0', R_1', R_2']$ is at most $2^{|R_0'| + |R_1'|}$ by the size invariant.
    
    For simplicity of the complexity analysis, we may employ the coarser upper bound $|\hat{\mathcal{S}}_{t_1}[D', R_0', R_1', R_2']| \leq 2^{|R_0'| + |R_1'| + |R_2'|}$. This bound suffices to ensure a single exponential time algorithm.
    Similarly, we have that $|\hat {\cal S}_{t_2}[D'', R_0'', R_1'', R_2'']|\leq 2^{|R_0''| + |R_1''| + |R_2''|}$.
    Since $D = D' = D''$, we have that $|R_0| + |R_1| + |R_2| = |R_0'| + |R_1'| + |R_2'| = |R_0''| + |R_1''| + |R_2''|$.
    The size of $\hat {\cal S}_t[D, R_0, R_1, R_2]$ is bounded by 
    \[
        2^{|R_1|} \cdot 3^{|R_2|} \cdot 2^{|R_0'| + |R_1'| + |R_2'|}\cdot 2^{|R_0''| + |R_1''| + |R_2''|} = 2^{|R_1|}\cdot  3^{|R_2|} \cdot  4^{|R_0| + |R_1| + |R_2|}.
    \]
    
    Now, we consider the running time of applying Lemma \ref{Lemma:shrinking} for all possible partition $(D, R_0, R_1, R_2)$ of $X_t$. which is bounded by 
    \[
        \begin{split}
            & O\left(\sum_{i = 1}^{\mathbf{tw} + 1} {\mathbf{tw} + 1\choose i} \sum_{j = 1}^{i} {i\choose j} \sum_{k = 1}^{j} {j\choose k}  4^i 3^{i - j} 2^{j - k} (2^{j(\omega - 1)}\mathbf{tw}^{O(1)}n)\right) \\
            = & O\left(\sum_{i = 1}^{\mathbf{tw} + 1} {\mathbf{tw} + 1\choose i} 12^i  \sum_{j = 1}^{i} {i\choose j} (2^{\omega}/3)^j \sum_{k = 1}^{j} {j\choose k}  (1/2)^k \mathbf{tw}^{O(1)}n\right) \\
            = & O\left(\sum_{i = 1}^{\mathbf{tw} + 1} {\mathbf{tw} + 1\choose i} 12^i  \sum_{j = 1}^{i} {i\choose j} (2^{\omega}/3)^j \cdot (3/2)^j \mathbf{tw}^{O(1)}n\right) \\
            = & O\left(\sum_{i = 1}^{\mathbf{tw} + 1} {\mathbf{tw} + 1\choose i} 12^i \cdot  (1 + 2^{\omega - 1})^i  \mathbf{tw}^{O(1)}n\right) \\
            = & O((13 + 3\cdot 2^{\omega + 1})^{\mathbf{tw}} \mathbf{tw}^{O(1)}n).
        \end{split}
    \]

    Thus, the whole running time of a join node is bounded by $O((13 + 3\cdot 2^{\omega + 1})^{\mathbf{tw}} \mathbf{tw}^{O(1)}n)$.
    \vspace{2mm}

    Since the number of node in a nice tree decomposition is bounded by $O(n)$, the whole algorithm takes $O((13 + 3\cdot 2^{\omega + 1})^{\mathbf{tw}} \mathbf{tw}^{O(1)}n)$ time.
    However, it is not necessary to compute the forests and the weights corresponding to the partial solutions in each node.
    Since the size of the forest is at most $\mathbf{tw} + 1$, we can maintain all connectivity graphs with weights during the dynamic programming process in time $\mathbf{tw}^{O(1)}$.
    To get linear time algorithms, we only store connectivity graphs and update them instead of storing partial solutions.
    For join node, it is possible to judge the combined partition solution is feasible only by combining the connectivity graphs.
    
    Thus, we can reduce the running time to $O((13 + 3\cdot 2^{\omega + 1})^{\mathbf{tw}} \mathbf{tw}^{O(1)}n)$.
    
    For pathwidth, since there is no join node in a nice path decomposition, our algorithm runs in $O((2 + 2^{\omega + 1})^{\mathbf{pw}}\mathbf{pw}^{O(1)}n)$ time.
    This theorem holds.
\end{proof}

\section{Conclusion}
In this paper, we present the first deterministic single exponential time algorithms for \textsc{Co-Path Packing} and \textsc{Co-Path Set} parameterized by treewidth. Specifically,
given a tree decomposition of width $\mathbf{tw}$,
we demonstrate that \textsc{Co-Path Packing} can be solved in $O((13 + 3\cdot 2^{\omega + 1})^{\mathbf{tw}} \mathbf{tw}^{O(1)}n)$ time, while \textsc{Co-Path Set} can be solved in $O((12 + 3 \cdot 2^{\omega + 1})^{\mathbf{tw}} \mathbf{tw}^{O(1)}n)$ time.
While there may be opportunities to further optimize these running time bounds -- particularly by refining the analysis of join node -- our results resolve the long-standing question regarding the existence of deterministic single exponential algorithms for these problems. Moving forward, it would be of significant interest to establish non-trivial lower bounds for these problems under the Strong Exponential Time Hypothesis (SETH).

\bibliographystyle{splncs04}
\bibliography{cPPST}

@book{cygan2015parameterized,
  title={Parameterized algorithms},
  author={Cygan, Marek and Fomin, Fedor V and Kowalik, {\L}ukasz and Lokshtanov, Daniel and Marx, D{\'a}niel and Pilipczuk, Marcin and Pilipczuk, Micha{\l} and Saurabh, Saket},
  volume={5},
  number={4},
  year={2015},
  publisher={Springer}
}

@inproceedings{chen2010linear,
  title={A linear kernel for co-path/cycle packing},
  author={Chen, Zhi-Zhong and Fellows, Michael and Fu, Bin and Jiang, Haitao and Liu, Yang and Wang, Lusheng and Zhu, Binhai},
  booktitle={Algorithmic Aspects in Information and Management: 6th International Conference, AAIM 2010, Weihai, China, July 19-21, 2010. Proceedings 6},
  pages={90--102},
  year={2010},
  organization={Springer}
}

@article{feng2015randomized,
  title={Randomized parameterized algorithms for $P_2$-packing and co-path packing problems},
  author={Feng, Qilong and Wang, Jianxin and Li, Shaohua and Chen, Jianer},
  journal={Journal of Combinatorial Optimization},
  volume={29},
  pages={125--140},
  year={2015},
  publisher={Springer}
}

@article{tsur2022faster,
  title={Faster deterministic algorithms for Co-path Packing and Co-path/cycle Packing},
  author={Tsur, Dekel},
  journal={Journal of Combinatorial Optimization},
  volume={44},
  number={5},
  pages={3701--3710},
  year={2022},
  publisher={Springer}
}

@inproceedings{sullivan2017fast,
  title={A Fast Parameterized Algorithm for Co-Path Set},
  author={Sullivan, Blair D and van der Poel, Andrew},
  booktitle={11th International Symposium on Parameterized and Exact Computation (IPEC 2016)},
  pages={28--1},
  year={2017},
  organization={Schloss Dagstuhl--Leibniz-Zentrum f{\"u}r Informatik}
}

@inproceedings{cygan2011solving,
  title={Solving connectivity problems parameterized by treewidth in single exponential time},
  author={Cygan, Marek and Nederlof, Jesper and Pilipczuk, Marcin and Pilipczuk, Michal and van Rooij, Joham MM and Wojtaszczyk, Jakub Onufry},
  booktitle={2011 IEEE 52nd Annual Symposium on Foundations of Computer Science},
  pages={150--159},
  year={2011},
  organization={IEEE}
}

@article{liu2026solving,
  title={Solving co-path/cycle packing and co-path packing faster than 3k},
  author={Liu, Yuxi and Xiao, Mingyu},
  journal={Information and Computation},
  pages={105406},
  year={2026},
  publisher={Elsevier}
}

@book{oxley2006matroid,
  title={Matroid theory},
  author={Oxley, James G},
  volume={3},
  year={2006},
  publisher={Oxford University Press, USA}
}

@article{DBLP:journals/jacm/FominLPS16,
  title={Efficient computation of representative families with applications in parameterized and exact algorithms},
  author={Fomin, Fedor V and Lokshtanov, Daniel and Panolan, Fahad and Saurabh, Saket},
  journal={Journal of the ACM (JACM)},
  volume={63},
  number={4},
  pages={1--60},
  year={2016},
  publisher={ACM New York, NY, USA}
}

@article{DBLP:journals/dam/Fernau08,
  author       = {Henning Fernau},
  title        = {Parameterized algorithmics for linear arrangement problems},
  journal      = {Discret. Appl. Math.},
  volume       = {156},
  number       = {17},
  pages        = {3166--3177},
  year         = {2008}
}

@article{DBLP:journals/jco/ZhangJZ14,
    title={Radiation hybrid map construction problem parameterized},
  author={Zhang, Chihao and Jiang, Haitao and Zhu, Binhai},
  journal={Journal of Combinatorial Optimization},
  volume={27},
  number={1},
  pages={3--13},
  year={2014},
  publisher={Springer}
}

@article{DBLP:journals/jco/FengZW16,
  author       = {Qilong Feng and
                  Qian Zhou and
                  Jianxin Wang},
  title        = {Kernelization and randomized Parameterized algorithms for Co-path
                  Set problem},
  journal      = {J. Comb. Optim.},
  volume       = {32},
  number       = {1},
  pages        = {67--78},
  year         = {2016}
}

@article{DBLP:journals/ipl/Tsur23a,
  title={Faster deterministic algorithm for co-path set},
  author={Tsur, Dekel},
  journal={Information Processing Letters},
  volume={180},
  pages={106335},
  year={2023},
  publisher={Elsevier}
}

@inproceedings{DBLP:conf/faw/FengZL14,
  title={Randomized parameterized algorithms for co-path set problem},
  author={Feng, Qilong and Zhou, Qian and Li, Shaohua},
  booktitle={International Workshop on Frontiers in Algorithmics},
  pages={82--93},
  year={2014},
  organization={Springer}
}

@article{DBLP:journals/iandc/BodlaenderCKN15,
  title={Deterministic single exponential time algorithms for connectivity problems parameterized by treewidth},
  author={Bodlaender, Hans L and Cygan, Marek and Kratsch, Stefan and Nederlof, Jesper},
  journal={Information and Computation},
  volume={243},
  pages={86--111},
  year={2015},
  publisher={Elsevier}
}

@article{DBLP:journals/ploscb/ChauveT08,
  title={A methodological framework for the reconstruction of contiguous regions of ancestral genomes and its application to mammalian genomes},
  author={Chauve, Cedric and Tannier, Eric},
  journal={PLoS computational biology},
  volume={4},
  number={11},
  pages={e1000234},
  year={2008},
  publisher={Public Library of Science San Francisco, USA}
}

@article{cox1990radiation,
  title={Radiation hybrid mapping: a somatic cell genetic method for constructing high-resolution maps of mammalian chromosomes},
  author={Cox, David R and Burmeister, Margit and Price, E Roydon and Kim, Suwon and Myers, Richard M},
  journal={Science},
  volume={250},
  number={4978},
  pages={245--250},
  year={1990},
  publisher={American Association for the Advancement of Science}
}

@article{richard1991radiation,
  title={A radiation hybrid map of the proximal long arm of human chromosome 11 containing the multiple endocrine neoplasia type 1 (MEN-1) and bcl-1 disease loci},
  author={Richard 3rd, CW and Withers, DA and Meeker, TC and Maurer, S and Evans, GA and Myers, RM and Cox, DR},
  journal={American journal of human genetics},
  volume={49},
  number={6},
  pages={1189},
  year={1991}
}

@inproceedings{slonim1997building,
  title={Building human genome maps with radiation hybrids},
  author={Slonim, Donna and Kruglyak, Leonid and Stein, Lincoln and Lander, Eric},
  booktitle={Proceedings of the first annual international conference on Computational molecular biology},
  pages={277--286},
  year={1997}
}

@inproceedings{alman2025more,
  title={More asymmetry yields faster matrix multiplication},
  author={Alman, Josh and Duan, Ran and Williams, Virginia Vassilevska and Xu, Yinzhan and Xu, Zixuan and Zhou, Renfei},
  booktitle={Proceedings of the 2025 Annual ACM-SIAM Symposium on Discrete Algorithms (SODA)},
  pages={2005--2039},
  year={2025},
  organization={SIAM}
}

\end{document}